\documentclass[journal]{IEEEtran}

\usepackage{authblk} 
\usepackage{scalerel}
\usepackage{times}
\usepackage{relsize}

\usepackage{enumitem}
\usepackage{float}
\usepackage{mathtools}
\usepackage{graphicx}
\usepackage{graphics}
\usepackage[dvips]{epsfig}
\usepackage{subfig}
\usepackage{verbatim}
\usepackage{graphicx,float,latexsym,amssymb,amsfonts,amsmath,amstext,times,epsfig,setspace}
\usepackage{algorithm}
\usepackage[utf8]{inputenc}
\usepackage[noend]{algpseudocode}
\usepackage{fancyhdr}
\usepackage{stfloats}
\usepackage{tabu}
\usepackage{booktabs}
\usepackage{amsmath}
\usepackage[noend]{algpseudocode}
\renewcommand{\algorithmiccomment}[1]{\bgroup\hfill\scriptsize//~#1\egroup}

\newenvironment{proof}[1][Proof]{\begin{trivlist}
\item[\hskip \labelsep {\bfseries #1}]}{\end{trivlist}}

\newcommand{\qed}{\nobreak \ifvmode \relax \else
      \ifdim\lastskip<1.5em \hskip-\lastskip
      \hskip1.5em plus0em minus0.5em \fi \nobreak
      \vrule height0.75em width0.5em depth0.25em\fi}
      
\newtheorem{theorem}{Theorem}[section]
\newtheorem{lemma}[theorem]{Lemma}
\newtheorem{proposition}[theorem]{Proposition}
\newtheorem{corollary}[theorem]{Corollary}
\newtheorem{definition}[theorem]{Definition}

\DeclareMathOperator*{\E}{\mathbb{E}}

\newcommand{\client}{client}

\newcommand{\enode}{encoding node}

\newcommand{\Kmax}{{K_\mathrm{max}}}

\newcommand{\node}{node}

\newcommand{\object}{object}

\newcommand{\dobjsz}{object size}

\newcommand{\ssize}{symbol size}

\newcommand{\ldn}{Liquid Data Networking}
\newcommand{\LDN}{LDN}

\newcommand{\sobject}{stream object}
\newcommand{\Sobject}{Stream object}
\newcommand{\SOPI}{SOPI}

\newcommand{\edata}{encoded data}

\newcommand{\fcode}{fountain code}

\newcommand{\ecode}{erasure code}

\newcommand{\fencoding}{fountain encoding}

\newcommand{\fdecoding}{fountain decoding}

\begin{document}
\pagenumbering{gobble}

\title{\SOPI\ design and analysis for \LDN}
\thispagestyle{fancy}

\author{Michael~Luby\thanks{The research described in this paper is supported by NSF EAGER 1936572.}}
\affil{ICSI and BitRipple, Inc.}

\maketitle

\begin{abstract}
\emph{\ldn\ (\LDN)}, introduced in \cite{byersluby2020ldn}, 
is an ICN architecture that is designed to enable 
the benefits of erasure-code enabled object delivery.
A primary contribution of \cite{byersluby2020ldn}
is the introduction of \SOPI s, which enables
\client s to concurrently download \edata\ for the same \object\ 
from multiple edge \node s, optimizes caching efficiency, and enables seamless mobility.
This paper provides an enhanced design and
analysis of \SOPI s.
\end{abstract}

\IEEEpeerreviewmaketitle


\section{Introduction}

The paper~\cite{byersluby2020ldn}, inspired by 
\emph{\fcode s}~\cite{dfapproach},
\cite{RFC6330}, \cite{RaptorQpap}, \cite{codornices}, introduces
{\em \ldn\ (\LDN)}, an ICN architecture that is designed to enable the benefits of 
erasure-code enabled \object\ delivery.
A primary contribution of \cite{byersluby2020ldn}
is the introduction of \emph{\sobject\ permutation identifiers (\SOPI s)}, 
which provides a simple and efficient download coordination mechanism:
\client s can concurrently download \edata\ for the same \object\ from multiple edge \node s,
caching efficiency is optimized, and seamless mobility is enabled.
This paper provides an enhanced design and analysis of \SOPI s, and inherits the
terminology and notation of~\cite{byersluby2020ldn}.

\section{\SOPI\ and \Sobject}
\label{stream sec}

\SOPI s and \emph{\sobject s} are fundamental to the design of \LDN:
they enable a diversity of \edata\ to be available
for download within the network for each \object, 
while at the same time ensuring that different \client s
request the same \edata\ for an \object\ 
from the same neighboring \enode.
The \SOPI\ design allows the \client\ decision 
of which \edata\ to request for an \object\ to be simple, robust, and efficient.

An \object\ is partitioned into $K$ source symbols, where $K$ is the \dobjsz\ divided 
by the \ssize, and the \ssize\ is typically chosen to fit into a packet payload.
\Sobject s can be described in terms of an \ecode\ with
optimal recovery properties and with the property 
that $N$ \edata\ symbols can be generated from any \object, 
where $N$ is a large prime number, and in particular
$N >> \Kmax$, where $\Kmax$ is the maximium
number of source symbols in any \object.

A \sobject\ for an \object\ consists of all of the possible \edata\ 
that can be generated for the \object\ in a specified order.  
The essential idea is that different \sobject s specify completely 
different orderings of the available \edata\ for an \object, 
and thus a \client\ can simply request prefixes of different 
\sobject s for an \object\ to receive non-overlapping 
\edata.

A \emph{\sobject\ permutation identifier (\SOPI)} specifies the 
ordering of \edata\ symbols for a \sobject.
A \SOPI\  $P$ identifies a permutation $\pi(P)$ of the 
$N$ available symbol IDs.
Thus, a \sobject\ identifier $(D,P)$ is the combination of the identifier $D$ 
of the \object\ from which it is generated and a \SOPI\ $P$.
We consider \SOPI s of the form $P=(A,B)$, 
where $A \in \{0, 1, 2, \ldots, N-1\}$,  
$B \in \{1, 2, 3, \ldots, N-1\}$.
Then, $P=(A,B)$ defines the permutation of symbol IDs
\[ \pi(P) = \{A, A+B, A+2 \cdot B, \ldots, A+(N-1) \cdot B \}, \]
where each term is taken mod $N$, i.e.,  
\[ \pi(P)[i] = A + i \cdot B \mod N \]
for any position $i \in \{ 0,\ldots,N-1\}$.

\subsection{RaptorQ \SOPI\ implementation}
\label{raptorq sopi sec}

For the implementation of the RaptorQ code~\cite{RFC6330} described at~\cite{codornices}, there are $2^{31}$ 
possible symbols of \edata\ for an \object.
Since $2^{31}-1$ is a prime number (a Mersenne prime),
$N=2^{31}-1$ can be used for the implementation~\cite{codornices} 
of the RaptorQ code.

Since $N$ is one less than a power of two, it is 
straightforward to compute $C = A + i \cdot B \mod N$, e.g.,
\begin{itemize}
    \item Compute $D = i \cdot B$.
    \item Let $D_0$ be the $31$ least significant bits of $D$.
    \item Let $D_1$ be the next $31$ bits of $D$.
    \item Let $C = A + D_0 + D_1$.
    \item If $C > N$ then $C  = C - N$.
    \item If $C > N$ then $C  = C - N$.
\end{itemize}

\section{Random \SOPI\ sets} 
\label{random analysis sec}

We analyze the properties of collections of \SOPI s,
where each \SOPI\ in the collection is randomly and
independently chosen.

\begin{proposition}
\label{pairwise prop}
For any pair of positions \[i_0 \in \{0,1,\ldots,N-1\}, i_1 \in \{0,1,\ldots,N-1\} - \{i_0\},\] 
and for any pair of symbol IDs 
\[j_0 \in \{0,1,\ldots,N-1\}, j_1 \in \{0,1,\ldots,N-1\} - \{j_0\},\]
there is a unique $P=(A,B)$
such that \[\pi(P)[i_0] = j_0 \mbox{ and } \pi(P)[i_1] = j_1.\]
\qed 
\end{proposition}

Proposition~\ref{pairwise prop} implies that the pair
of symbol IDs in any pair of positions within 
the permutation are random with respect to a random \SOPI\ $P=(A,B)$.

Suppose an \object\ $D$ is composed of a single source block
with $K$ source symbols.
A \client\ will download \edata\ from prefixes of multiple
\sobject s of an \object\ $D$ in order to recover the \object.
We would like to show that a \client\ doesn't receive many
duplicate symbols when downloading from multiple \sobject s
with randomly chosen \SOPI s.  Suppose the \client\ downloads
from some number $s$ of different \sobject s, with corresponding
randomly chosen
\SOPI s $P_0, P_1, \ldots, P_{s-1}$.  Suppose the 
\client\ receives $M_0$ symbols from the \sobject\ $(P_0,D)$,
$M_1$ symbols from the \sobject\ $(P_1,D)$, etc.  where
\[ \sum_{i=0}^{s-1} M_i = M. \]

The expected number of distinct symbol IDs among $M$ is at least
\[M -\frac{(M-1)^2}{2 \cdot N},\] 
and thus if $K^2 \le 2 \cdot N$ then 
\[M =K \cdot \left(1+ \frac{K}{2 \cdot N}\right) \le K+1\] 
symbols on average
are enough to receive $K$ symbols with at least $K$ distinct symbol IDs.  
Theorem~\ref{dup theorem} below provides bounds on the probability 
that $M$ received symbols have at least $K$ distinct symbol IDs.

\begin{theorem}
\label{dup theorem}
Consider an \object\ $D$ composed of a single source block with $K$ source symbols. 
Suppose at least \[M = \frac{K}{1-\delta}\] symbols have been received in total 
from $s \ge 1$ \sobject s\ 
\[(P_0,D),(P_1,D),\ldots, (P_{s-1},D),\] 
where $\delta > 0$ and $M^2 \le 2 \cdot N$.
Then, 
\begin{equation}
\label{theorem eq}
\Pr[ D \mbox{ not recoverable from the $M$ symbols} ] \le \frac{1}{\delta^2 \cdot N}
\end{equation}
with respect to the random variables 
\[ P_0 = (A_0,B_0) ,P_1 = (A_1,B_1),\ldots, P_{s-1} = (A_{s-1},B_{s-1}). \]
\end{theorem}

\begin{proof}
Let $Y$ be the random variable that is the number 
of unique symbol IDs among $M$ 
received symbols for an \object\ $D$ with respect to the random variables 
\[ P_0 = (A_0,B_0) ,P_1 = (A_1,B_1),\ldots, P_{s-1} = (A_{s-1},B_{s-1}). \]
The left-hand side of Inequality~(\ref{theorem eq}) is equivalent to 
\begin{equation}
  \Pr[ Y < K]  = \Pr[ Y < (1-\delta) \cdot M]
\end{equation}
and thus we need to prove that
\begin{equation}
\label{toprove eqn}
    \Pr[Y < (1-\delta) \cdot M] \le \frac{1}{\delta^2 \cdot N}.
\end{equation}

Each of the $M$ received symbols has a symbol position within the \sobject\ from which it is received, 
and a symbol ID determined by the \sobject\ \SOPI\ and the symbol position within the \sobject.
We associate a unique index within $\{0, 1, \ldots, M-1 \}$ with each of 
the $M$ (\sobject, symbol position) pairs of received symbols.  
For $i \in \{0, 1, \ldots, M-1 \}, j \in \{0, 1, \ldots, M-1 \} - \{i\}$,
we let $X_{i,j} = 1$ if the symbol position within the \sobject\ indexed by $i$ is mapped to
the same symbol ID as the symbol position within the \sobject\ indexed by $j$, 
and $X_{i,j} = 0$ if the symbol position within the \sobject\ indexed by $i$ is mapped to a
different symbol ID than the symbol position within the \sobject\ indexed by $j$.  
Thus,  $X_{i,j} = 1$ if and only if the symbol IDs mapped to by $i$ and $j$ are duplicates.
Then,
\begin{align}
  \Pr[Y < (1-\delta) \cdot M] & = \Pr[M-Y > \delta \cdot M] \nonumber \\
  & \le \Pr\left[\sum_{i,j > i} X_{i,j} > \delta \cdot M\right] \label{x2 eq} \\
  & = \Pr\left[\left(\sum_{i, j > i} X_{i,j}\right)^2 > \delta^2 \cdot M^2\right] \nonumber \\
  & \le \frac{\E\left[\left(\sum_{i,j > i} X_{i,j}\right)^2\right]}{\delta^2 \cdot M^2}, \label{markov eq}
\end{align}
where Inequality~(\ref{x2 eq}) follows since $M-Y$ is the number of duplicate symbol IDs
and if the symbol IDs mapped to by $i$ and $j$ are the same then $X_{i,j}=1$,
and Inequality~(\ref{markov eq}) follows from Markov's inequality.

Expanding out terms,
\begin{align}
  \E\left[\left(\sum_{i,j > i} X_{i,j}\right)^2\right]
  & = \E\left[\sum_{i,j > i} X_{i,j}^2 \right] \label{xxx eq}\\
  & +
  \E\left[ \sum_{i,j>i}\mathop{\sum_{i',j'>i'}}_{(i',j') \not= (i,j)} X_{i,j} \cdot X_{i',j'} \right] \label{xxxx eq}
\end{align}
If $i$ and $j$ index different symbol positions within the same \sobject\ then $\E[X^2_{i,j}]=0$
since the \SOPI\ for the \sobject\ defines a permutation of the symbol IDs 
and thus cannot map $i$ and $j$  to the same symbol ID.  
If $i$ and $j$ index symbol positions within two different \sobject s
then $\E[X^2_{i,j}]=\frac{1}{N}$ since the \SOPI s for the two 
different \sobject s are chosen independently.  
Thus,
\begin{equation}
\label{right xxx eq}
\E\left[\sum_{i,j > i} X_{i,j}^2 \right]  \le \frac{\binom{M}{2}}{N} \le \frac{M^2}{2 \cdot N}.   
\end{equation}

Similarly, if $i$ and $j$ index different symbol positions within the same \sobject\ 
or $i'$ and $j'$ index different symbol positions within the same \sobject\ then 
$\E\left[X_{i,j} \cdot X_{i',j'} \right] = 0$.
If $i$ and $i'$ index different symbol positions within one \sobject\ 
and $j$ and $j'$ index different symbols positions within a second \sobject\ 
then $\E\left[X_{i,j} \cdot X_{i',j'} \right] = \frac{1}{(N-1) \cdot N}$ 
by Proposition~\ref{pairwise prop}. If $i$ and $i'$ index different symbol
positions within one \sobject\ and $j$ and $j'$ index the same symbol position
within a second \sobject\ then $\E\left[X_{i,j} \cdot X_{i',j'} \right] = 0$
since $i$ and $i'$ cannot map to the same symbol ID.
If $i$ and $i'$ index different symbol positions within one \sobject\ and $j$ indexes 
a symbol position in a second \sobject\ and $j'$ indexes a position within a third \sobject\ 
then $\E\left[X_{i,j} \cdot X_{i',j'} \right] = \frac{1}{N^2}$, since if $j$ and $j'$ map
to the same symbol ID then $i$ and $i'$ cannot both map to that symbol ID,
and if $j$ and $j'$ map to different symbol IDs (with probability $\frac{N-1}{N}$)
then $i$ and $j$ map to the same symbol ID and $i'$ and $j'$ map to the same
symbol ID with probability $\frac{1}{(N-1) \cdot N}$
from Proposition~\ref{pairwise prop}.
If $i$, $j$, $i'$ and $j'$ index symbol positions in four different \sobject s
then $\E\left[X_{i,j} \cdot X_{i',j'} \right] = \frac{1}{N^2}$, since each of
the four \SOPI s are chosen randomly and independently of one another.
All other cases are variants of these cases.  Thus,
\begin{equation}
\label{xxxxxx eq}
\E\left[ \sum_{i,j>i}\mathop{\sum_{i',j'>i'}}_{(i',j') \not= (i,j)} X_{i,j} \cdot X_{i',j'} \right]
\le \frac{\binom{M}{2}^2}{(N-1) \cdot N} \le \frac{M^4}{4 \cdot N^2}.   
\end{equation}
From Inequalities~(\ref{right xxx eq}) and~(\ref{xxxxxx eq}) it follows that
\begin{align}
    \E\left[\left(\sum_{i,j > i} X_{i,j}\right)^2\right]  & \le 
    \frac{M^2}{2 \cdot N} \cdot \left(1 + \frac{M^2}{2 \cdot N}\right) \\
    & \le \frac{M^2}{N}, \label{sqrtM eq}
\end{align}
where Inequality~(\ref{sqrtM eq}) follows since $M^2 \le 2 \cdot N$.
Combining Inequality~(\ref{sqrtM eq}) with Inequality~(\ref{markov eq})
proves Inequality~(\ref{toprove eqn}).
\qed
\end{proof}

\subsection{Random \SOPI\ set examples}
\label{raptorq random sec}
For the RaptorQ code specified in~\cite{RFC6330}, the maximum supported 
number of source symbols per source block is $\Kmax = 56,403$, and $N = 2^{31}-1$ 
is a good choice as described in Section~\ref{raptorq sopi sec}.
With $M=65,535$, the condition $M^2 \le 2 \cdot N$ is satisfied, 
and $\Kmax \le 0.86 \cdot M = (1-\delta) \cdot M$ for $\delta = 0.14$.  
Thus, Theorem~\ref{dup theorem} holds for RaptorQ
for all supported source block sizes and for any $\delta \le 0.14$.

On average at most a $0.00002$ fraction of the symbol IDs of 
received \edata\ symbols from prefixes of \sobject s will be duplicates
when the total number of received symbols is up to $M$.  However, 
stronger bounds on the probability of the number of duplicates being
above a certain bound are of importance,

For example, setting $\delta = 0.01$, Theorem~\ref{dup theorem} shows that
an \object\ $D$ composed of a single source block of $K$ source symbols
can be recovered with probability at least 
\[ 1-\frac{1}{\delta^2 \cdot N} = 0.999995\]
from $M = K/0.99 \approx 1.01 \cdot K$ received symbols in total
from different \sobject s.

As another example, setting $\delta = 0.1$, Theorem~\ref{dup theorem} shows that
an \object\ $D$ composed of a single source block of $K$ source symbols
can be recovered with probability at least 
\[ 1-\frac{1}{\delta^2 \cdot N} = 0.99999995\]
from $M = K/0.9 \approx 1.11 \cdot K$ received symbols in total
from different \sobject s.

The analysis of Theorem~\ref{dup theorem}
is not tight, and thus these bounds are conservative.

\section{Designed \SOPI\ sets}

Although the analysis in Section~\ref{random analysis sec} shows there is
little chance of there being a significant fraction of duplicates among
received symbols from multiple \sobject s, there is still a chance that
there is a significant fraction of duplicates. As an extreme example,
if $P_0$ and $P_1$ are identical then the overlap between prefixes
of length $K/2$ is $K/2$, i.e., all $K/2$ symbols are duplicates.
As a somewhat less extreme example, if $P_0 = (0,1)$ 
and $P_1 = (A,1)$ where $A$ is randomly chosen, then for prefixes of
length $K/2$ there is a $1-K/N$ chance there is no overlap between the
prefixes, but with the remaining probability $K/N$ there are on average $K/4$ duplicates.

In this section, we provide a deterministic design of a large set 
$\mathcal{P}$ of \SOPI s for which we can provide guaranteed 
upper bounds on the number of symbol ID duplicates there are 
in prefixes of permutations defined by \SOPI s from $\mathcal{P}$. 

The design has two parts:
\begin{itemize}
    \item Design a set $\mathcal{B} \subset \{1,\ldots,N-1\}$.
    \item For each $B \in \mathcal{B}$, design a set $\mathcal{A}_B \subset \{0,\ldots,N-1\}$. 
\end{itemize}
Then, the designed set of \SOPI s is the set 
\[ \mathcal{P} = \{P = (A,B): B \in \mathcal{B}, A \in \mathcal{A}_B \}. \]

\subsection{Interactions between \emph{B} values}

The design of set $\mathcal{B}$ is based on interactions between different $B$ values. 
The following is at the heart of this analysis.

\begin{definition}
\label{matches defn}
Let $B_0, B_1 \in \{1,\ldots,N\}$ such that $B_0 \not= B_1$.
For any pair of integers $(d_0,d_1)$,
we say $(d_0,d_1)$ \emph{matches with respect to $(B_0,B_1)$}
if 
\[d_0 \cdot B_0 \mod N = d_1 \cdot B_1 \mod N. \]
\qed
\end{definition}
The importance of Definition~\ref{matches defn} is that,
for any $A_0, A_1 \in \{0,\ldots,N\}$,
if the symbol ID at position $p_0$ in the permutation defined by
$P_0 = (A_0,B_0)$ is equal to the symbol ID at position $p_1$ 
in the permutation defined by $P_1 = (A_1,B_1)$
and $(d_0,d_1)$ matches with respect to $(B_0,B_1)$, 
then the symbol ID at position $p_0 + d_0 \mod N$ 
with respect to $P_0$ is equal to the  
symbol ID at position $p_1 + d_1 \mod N$ with respect to $P_1$.

It is easy to verify that if $(d_0,d_1)$ matches 
with respect to $(B_0,B_1)$ and $(d'_0,d'_1)$ matches 
with respect to $(B_0,B_1)$ then
\[ (d_0,d_1) + (d'_0,d'_1) = (d_0+d'_0,d_1+d'_1) \]
matches with respect to $(B_0,B_1)$. 
For example, if $(d_0,d_1)$ matches 
with respect to $(B_0,B_1)$ then, for any integer $i$,
\[ i \cdot (d_0,d_1) = (i \cdot d_0, i \cdot d_1) \]
matches with respect to $(B_0,B_1)$.

Let $M$ be an upper bound on the aggregate number of symbols
downloaded by a \client\ from prefixes of \sobject s 
for an \object. We impose the condition that
\begin{equation}
    \label{M upper eq}
    M^2 < N/2.
\end{equation}

\begin{definition}
\label{M set defn}
The set $D$ is defined as
\[ D = \{-M+1,\ldots, -1 \} \cup \{ 1,\ldots, M-1 \}.\]
\qed
\end{definition}

\begin{lemma}
\label{main distance lemma}
For any $B_0 \in \{1,\ldots,N\}$ and $B_1 \in \{1,\ldots,N\}$, one of the following 
two possibilities holds:
\begin{enumerate}
    \item \label{large dist eq}
    For all $d_0 \in D$, $d_1 \in D$, 
    $(d_0, d_1)$ does not match with respect to $(B_0, B_1)$.
The distance between $B_0$ and $B_1$ is defined to be $2 \cdot M$.
    \item \label{small dist eq}
    There is a $d_0 \in D$, $d_1 \in D$ such that
    $(d_0, d_1)$ matches with respect to $(B_0, B_1)$
    and, for any $d'_0 \in D$ and $d'_1 \in D$ such 
    that $(d'_0, d'_1)$ matches with respect to $(B_0,B_1)$,
    $(d'_0, d'_1)$ is an integer multiple of $(d_0,d_1)$.
The distance between $B_0$ and $B_1$ is defined to be $|d_0| + |d_1|$. 
\end{enumerate}

\end{lemma}

\begin{proof}
If (\ref{large dist eq}) holds the proof is complete. 
We need to show that if (\ref{large dist eq}) doesn't hold then (\ref{small dist eq}) holds.
If (\ref{large dist eq}) doesn't hold there is $d_0 \in D$ and $d_1 \in D$ such that
$(d_0,d_1)$ matches with respect to $(B_0,B_1)$ and $(d_0,d_1)$ is a minimal
pair, i.e., there is no $d'_0 \in D$ and $d'_1 \in D$ such that $(d'_0,d'_1)$ 
matches with respect to $(B_0,B_1)$ and $(d'_0,d'_1) = c \cdot (d_0,d_1)$,
where $0 < |c| < 1$.

Consider any $(d'_0,d'_1)$ with $d'_0 \in D$, $d'_1 \in D$
such that $(d'_0,d'_1)$ matches with respect to $(B_0,B_1)$.
We want to show that $(d'_0,d'_1)$ is an integer multiple of $(d_0,d_1)$.
Note that $(d'_0 \cdot d_0,d'_0 \cdot d_1)$ matches with respect to
$(B_0,B_1)$ and also $(d'_0 \cdot d_0, d'_1 \cdot d_0)$ matches with respect to
$(B_0,B_1)$. 

For any $A_0, A_1 \in \{0,\ldots,N\}$, 
consider a pair of positions $(p_0,p_1)$ where the symbol ID at
position $p_0$ of the permutation defined by $P_0 = (A_0,B_0)$ 
is the same as the symbol ID at position $p_1$ 
of the permutation defined by $P_1 = (A_1,B_1)$.
Then, the symbol ID at position 
$p_0 + d'_0 \cdot d_0 \mod N$ of $P_0$ is the same
as the symbol ID at position $p_1 + d'_0 \cdot d_1 \mod N$ 
of $P_1$ and also the same as the symbol ID 
at position $p_1 + d'_1 \cdot d_0 \mod N$ of $P_1$,
which implies that
\begin{equation}
\label{equal mod eqn}
p_1 + d'_0 \cdot d_1 = p_1 + d'_1 \cdot d_0 \mod N.
\end{equation}
Since 
\[-N/2 < -M^2 \le d'_0 \cdot d_1 \le M^2 < N/2\]
and
\[-N/2 < -M^2 \le d'_1 \cdot d_0 \le M^2 < N/2,\]
it follows that Equation~(\ref{equal mod eqn}) holds 
if and only if 
\[d'_0 \cdot d_1 = d'_1 \cdot d_0.\]
Thus, for some $c \not= 0$, 
\[(d'_0,d'_1) = c \cdot (d_0,d_1).\]
Since $(d_0,d_1)$ is a minimal pair, $|c| \ge 1$.
Then
\begin{align}
   (d''_0,d''_1) & = (d'_0, d'_1) 
- \lfloor c \rfloor \cdot (d_0, d_1) \nonumber \\
& = (c - \lfloor c \rfloor) \cdot (d_0, d_1) \nonumber
\end{align}
matches with respect to $(B_0,B_1)$,
and $d''_0 \in D$ and $d''_1 \in D$.
Since $0 \le c - \lfloor c \rfloor < 1$
and $(d_0,d_1)$ is a minimal pair,
it follows that $c - \lfloor c \rfloor = 0$ and thus (\ref{small dist eq}) holds
because $c$ is an integer.
\qed
\end{proof}

\begin{lemma}
\label{overlap lemma}
Suppose for $B_0, B_1 \in \{1,\ldots,N\}$
the distance between $B_0$ and $B_1$ is $d$.  
Then, for any $A_0, A_1 \in \{0,\ldots,N\}$
the number of distinct symbol IDs in the prefixes
of the permutations defined by \SOPI s 
\[ P_0 = (A_0,B_0), P_1 = (A_1,B_1),\] 
is at least \[ m - \left\lfloor \frac{m-2}{d} \right\rfloor-1, \]
where $m \le M$ is the total length of the pair of prefixes.
\end{lemma}

\begin{proof}
If $m=1$, or more generally one of the prefixes is of length zero, then
the number of distinct symbol IDs is $m$.  Otherwise, neither prefix
is of length zero.  If there are no symbol IDs in common between the
two prefixes then the number of distinct symbol IDs is $m$.

Suppose there are symbol IDs in common between the two prefixes.
Then there is a symbol ID at some position $p_0$ with respect to $P_0$
that is the same as a symbol ID at some position $p_1$ with respect to $P_1$.

If (\ref{large dist eq}) of Lemma~\ref{main distance lemma} holds,
i.e., there is no $d_0 \in D$, $d_1 \in D$ such that $(d_0,d_1)$ matches
with respect to $(B_0,B_1)$, then the distance between 
$B_0$ and $B_1$ is $d=2 \cdot M$ and $(p_0,p_1)$ is the only pair 
of positions where the symbol IDs are the same among the two prefixes, 
and thus there are 
\[ m-1 = m - \left\lfloor \frac{m-2}{2 \cdot M} \right\rfloor - 1\] 
distinct symbol IDs.

If (\ref{small dist eq}) of Lemma~\ref{main distance lemma} holds,
i.e., there is a $d_0 \in D$, $d_1 \in D$ such that $(d_0,d_1)$ matches
with respect to $(B_0,B_1)$ and all other pairs $(d'_0,d'_1)$ that match with respect
to $(B_0,B_1)$, where $d'_0 \in D$, $d'_1 \in D$, are integer multiples of $(d_0,d_1)$.
Without loss of generality, $d_0 > 0$ and $d_1 > 0$ (the case where $d_0 > 0$ and
$d_1 < 0$ is similar).  Then the distance between $B_0$ and $B_1$ is $d = d_0 + d_1$.
The following maximizes the number of symbol IDs that are 
duplicates between $P_0$ and $P_1$:
\begin{itemize}
    \item The symbol ID in position $0$ of $P_0$ is the same as the symbol ID in
position $0$ of $P_1$.
    \item Let $i = \lfloor \frac{m-2}{d} \rfloor$. 
    Then set $m_0$ and $m_1$ such that $m_0 \ge i \cdot d_0 + 1$ and 
    $m_1 \ge i \cdot d_1 + 1$.  (The remaining $r = m - i \cdot d - 2$
    length can be assigned to $m_0$ and $m_1$ arbitrarily to satisfy $m_0+m_1=m$.)
\end{itemize}
It can be verified that the number of pairs of symbol IDs that are the same
between the two prefixes is maximized, and that this number of pairs is
\[ i+1 = \left\lfloor \frac{m-2}{d} \right\rfloor+1.\]
\qed
\end{proof}

\begin{corollary}
\label{overlap cor}
Suppose for $B_0,B_1,\ldots, B_{s-1} \in \{1,\ldots,N\}$
the distance between $B_i$ and $B_j$ is at least $d$ 
for all $i,j \in \{0,\ldots,s-1\}$ where $i \not= j$.  
Then, for any $A_0, A_1,\ldots, A_{s-1} \in \{0,\ldots,N\}$
the number of distinct symbol IDs in the prefixes
of the permutations defined by \SOPI s 
\[ P_0 = (A_0,B_0), P_1 = (A_1,B_1), \ldots, P_{s-1} = (A_{s-1},B_{s-1}),\] 
is at least
\[ m - (s-1) \cdot \left( \frac{m-s}{d} + \frac{s}{2} \right), \]
where $m \le M$ is the total length of the $s$ prefixes.
\end{corollary}

\begin{proof}
Let $m_0,m_1,\ldots,m_{s-1}$ be the respective lengths of the $s$ prefixes,
where $m = \sum_{i=0}^{s-1} m_i$. From Lemma~\ref{overlap lemma},
for each $i\in \{0,\ldots,s-1\}$, 
$j \in \{ 0,\ldots,i-1\}$, the number of symbol IDs that are the same between
the permutation defined by the prefix of $P_i$ and $P_j$ is at most
\[\frac{(m_i-1)+(m_j-1)}{d}+1.\]  
Thus, the number of symbol IDs that are the same
between all pairs of the prefixes is at most
\begin{align}
& \left( \sum_{i=0}^{s-1} \frac{(s-1) \cdot (m_i-1)}{d}\right) + \frac{(s-1) \cdot s}{2} \\
& \le (s-1) \cdot \left( \frac{m-s}{d} + \frac{s}{2} \right). 
\end{align}
\qed
\end{proof}

Lemma~\ref{overlap lemma} and Corollary~\ref{overlap cor} provide worst case
lower bounds on the total number of distinct symbol IDs in prefixes.
The actual number of distinct symbol IDs in practice are much closer to
the total size $m$ of the prefixes than to the worst case lower bounds.

\subsection{Constructing a \SOPI\ set}

The following algorithm constructs a set $\mathcal{B}$ 
such that for $B_0 \in \mathcal{B}$, $B_1\in \mathcal{B} - \{B_0\}$,
the distance between $B_0$ and $B_1$ is at least $d$:
\begin{itemize}
    \item Initialize $\mathcal{B} = \emptyset$, $\mathcal{B}' = \{1,\ldots,N-1\}$.
    \item Repeat until $\mathcal{B}' = \emptyset$.
    \begin{itemize}
        \item Choose any $B \in \mathcal{B}'$
        \item Delete $B$ from $\mathcal{B}'$ and
        add $B$ to $\mathcal{B}$.
        \item For all $i \in D$, $j \in D$ such that $|i| + |j| < d$
        \begin{itemize}
            \item Delete $B' = i \cdot B \cdot j^{-1} \mod N$ from $\mathcal{B}'$.
        \end{itemize}
    \end{itemize}
\end{itemize}
Each time an element is added to $\mathcal{B}$, after accounting for
some symmetries such as 
\[i \cdot B \cdot j^{-1} \mod N = -i \cdot B \cdot (-j)^{-1} \mod N ,\]
the number of elements deleted from $\mathcal{B}'$ for each element added to
$\mathcal{B}$ is at most
\[(d-2)\cdot(d-3) +1 \le d^2. \] 
Thus, \[ |\mathcal{B}| \ge \frac{N-1}{d^2} \] 
when the algorithm completes.

The following algorithm constructs a set $\mathcal{A}_B$ 
for $B \in \mathcal{B}$ such that for $A_0 \in \mathcal{A}_B$, $A_1\in \mathcal{A}_B - \{A_0\}$, the symbol IDs of prefixes of length up to $M$ of the permutations defined by \SOPI\ $P_0 = (A_0,B)$ and $P_1 = (A_1,B)$ are all distinct.
as follows:
\begin{itemize}
    \item Initialize $\mathcal{A}_B = \emptyset$
    \item Choose $A$ randomly and add $A$ to $\mathcal{A}_B$.
    \item For $i = 1,\ldots,\lfloor N/M \rfloor -1$ do
    \begin{itemize}
        \item $A = A + M \cdot B \mod N$
        \item Add $A$ to $\mathcal{A}_B$.
    \end{itemize}
\end{itemize}
Thus, \[ |\mathcal{A}_B| \ge N/M-1 \] 
when the algorithm completes.

Overall, the set of \SOPI s
\[ \mathcal{P} = \{P = (A,B): B \in \mathcal{B}, A \in \mathcal{A}_B \} \]
is constructed. A \SOPI\ can be assigned to an \enode\ by choosing any 
$P=(A,B) \in \mathcal{P}$ that hasn't been assigned so far.
Overall,
\[ |\mathcal{P}| \ge \frac{N^2}{d^2 \cdot M}\]
\SOPI s can be assigned.

\subsection{Designed \SOPI\ set examples}

As an example, distance of at least $101$ between pairs of $B$ values
might be sufficient.  Lemma~\ref{overlap lemma} guarantees that there are at least
$K$ distinct symbol IDs among the prefixes of any pair of \SOPI s  
with $M = 1.01\cdot K+1$ symbol IDs in total.
The value of $M=30,000$ might be sufficient for practical use cases, 
where $M^2 \le N/2$ when $N = 2^{31}-1$.
In this case, the number of possible \SOPI s available is over $15$ billion.

As another example, distance of at least $1,000$ between pairs of $B$ values
might be desirable.  Corollary~\ref{overlap cor} guarantees that downloading from prefixes of up to $10$
\sobject s with different \SOPI s, where the total number of downloaded symbols
is at least $10,000$ and at most $30,000$, 
will have less than $1.5\%$ duplication in symbol IDs.
In this case, the number of possible \SOPI s available is over $150$ million.

\section{\SOPI\ distribution design}
\label{distribution sec}

It is not necessary that each \enode\ is assigned a unique \SOPI.
To enjoy the \SOPI\ design properties, it is sufficient that 
a \client\ avoids downloading \edata\ for an \object\ from 
two \enode s with the same assigned \SOPI, which the \client\ 
can easily avoid no matter how the \SOPI s are assigned to \enode s. 
For example, if a \client\ 
is supplied with the same \SOPI\ for two edge \enode s\ reachable
over different interfaces, the \client\ can simply send 
requests for \edata\ to only one of the two interfaces, 
and thus avoid receiving duplicate \edata.
Of course, the benefits of the \SOPI\ design are reduced in this
case, since the \client\ cannot effectively download \edata\ 
from both edge \enode s for the same \object. 

Thus, one would like to distribute \SOPI s to \enode s in such a
way that it is unlikely that a \client\ will receive 
the same \SOPI\ from two different edge \enode s.
One can conceptually create a graph of all \enode s in the internet, 
where there is an edge connecting two \enode s if it
is possible that a \client\ may download prefixes of \sobject s from both \enode s for the
same \object.  Then, one can think of each \SOPI\ $P \in \mathcal{P}$ as being a 
different color. Assigning colors to the graph in such a 
way that no edge has the same color at both endpoints 
provides a valid assignment of \SOPI s to \enode s.

It seems likely that such a graph can be colored using
a small number of colors, e.g., less than $100$, and perhaps
much less than $100$ if some edges in the graph as allowed to have
the same color.  Thus, the number of \SOPI s needed overall
is likely to be rather small.  This can ameliorate one of
the possible concerns with the \SOPI\ design, which is that
some information about which \client s are requesting
which \object s might be revealed by the requests for \edata\ 
containing \SOPI s percolating through the network: the amount
of information revealed is minimized if the number of \SOPI s
distributed is small.

\section{Large \object s}
\label{large sec}

If an \object\ is not too large then it can be treated as a single source block,
i.e., \fencoding\ and \fdecoding\ can be applied to the entire \object.
However, it becomes inefficient to apply \fencoding\ and \fdecoding\ to 
\object s that are larger, since typically the \object\ needs to fit into
working memory.  The required working memory is typically linear in the
source block size, i.e., the working memory to recover a block 
with $K$ source symbols is typically proportional to the \ssize\ times $K$.

Thus, for longer \object s, it is useful to have an algorithm that automatically
partitions the \object\ into multiple source blocks.  Such an algorithm is specified
in~\cite{RFC6330} in Section~4.3.  For \LDN\, the following is a suitable simplification 
of that algorithm.
Let the desired \ssize\ be $T$, and let $WS$ be the maximum
source block size for which \fencoding\ and \fdecoding\ is efficient,
where $\lfloor WS/T \rfloor \le 56,403$.
An \object\ $D$ of size $F$ can be automatically partitioned into source blocks
as follows:
\begin{itemize}
    \item $Kt = \lceil F/T \rceil$ (no. of source symbols in $D$).
    \item $\Kmax = \lfloor WS/T \rfloor$ (max no. symbols per source block).
    \item $Z = \lceil Kt / \Kmax \rceil$ (no. of source blocks to split $D$ into).
    \item $(KL, KS, ZL, ZS) = \mbox{Partition}[Kt, Z]$.
\end{itemize}
The output of Partition$[Kt, Z]$ is $(KL, KS, ZL, ZS)$, 
where $KL = \lceil Kt/Z \rceil$, $KS = \lfloor Kt/Z \rfloor$, 
$ZL = Kt - KS \cdot Z$, and $ZS = Z - ZL$.
The function Partition$[Kt,Z]$ partitions a
block of $Kt$ source symbols into $Z$ approximately equal-number of source symbols blocks.  
More specifically, it partitions $Kt$ into $ZL$ blocks, each with $KL$ source symbols, 
and $ZS$ blocks, each with $KS$ source symbols. 

Given the \ssize\ $T$, the maximum source block size $WS$, 
and the \dobjsz\ $F$, an \enode\ or \client\ can determine
the parameters $(KL, KS, ZL, ZS)$ as described above,
where $Z = ZL + ZS$ is the total number of source blocks.
If $Z=1$ then the \object\ is considered as a single source block
with $KL = K = \lceil F/T \rceil$ source symbols.

If $Z>1$ then \object\ $D$ is partitioned into $Z$ source blocks
as follows:
\begin{itemize}
    \item The initial portion of \object\ $D$ 
    of size $ZL \cdot KL \cdot T$ 
    is partitioned into $ZL$ source blocks,
    each consisting of $KL$ source symbols of size $T$.
    \item The remaining portion of \object\ $D$ of size $F - ZL \cdot KL \cdot T$ is partitioned into $ZS$ source blocks, each consisting of
    $KS$ source symbols of size $T$.
\end{itemize}

\subsection{Large \object\ extension of \SOPI}
\label{large sopi sec}

We extend the definitions of \SOPI\ and \sobject\  
to be applicable to large \object s, i.e., \object s that
are partitioned into $Z > 1$ source blocks.  The source block
structure for a large \object\ $D$ of size $F$ is determined as
described in Section~\ref{large sec} based on $F$, \ssize\ $T$
and a maximum source block size $WS$.

A \SOPI\ for the large \object\ $D$ is of the form $P = (A,B,C,D)$,
where $A \in \{0,1,\ldots,N-1\},$ $B \in \{1, 2, \ldots, N-1\}$,
$C \in \{ 0, 1, \ldots, N-1 \}$ and $D \in \{1, 2, \ldots, N-1\}$.
$A$ and $B$ are used similarly to Subsection~\ref{stream sec}
to identify a symbol generated from a source block,
and $C$ and $D$ are used to identify one of the $Z$ source blocks.
For convenience, the ranges of $C$ and $D$ are chosen to match
the ranges of $A$ and $B$, and thus large \object s with up to
$N-1$ source blocks can be supported.
For any $i \in \{ 0, 1, \ldots, Z \cdot (N-1) \}$, let:
\begin{itemize}
    \item $r = \lfloor i/Z \rfloor$
    \item $j = (A + r \cdot B) \mod N$
    \item $j' = (i + C + r \cdot D) \mod Z$
\end{itemize}
Then the symbol at position $i$ of the \sobject\ identified by \SOPI\ $P$ for \object\ $D$
is the symbol identified by $j$ from the source block identified by $j'$.

Each \SOPI\ $P$ defines a \sobject\ $(P,D)$, which
is a permutation defined by $P$ of the $Z \cdot N$ possible symbols that can
be generated for the $Z$ source blocks of the \object\ $D$. The permutation $P$ has the
property that each set of $Z$ consecutive symbols of the \sobject\ starting at a position
that is a multiple of $Z$ all have the same symbol ID but are from different source blocks,
i.e., the $Z$ consecutive symbols are from a cyclic shift of the $Z$ source blocks.

The values of $C$ and $D$ should be chosen randomly chosen.
This ensures that the source block sequence within each set 
of $Z$ consecutive symbols of the \sobject\ is a random cyclic
shift, and that different pairs of cyclic shifts are random.
This ensures that packet losses on average with respect 
to randomly chosen $C$ and $D$ affect each source block of a large
\object\ equally.  This is true even if packet loss patterns 
depend on the position of the symbol within the \sobject\ carried
in a packet, as long as the packet loss does not depend on the
values of $C$ and $D$.

\subsection{RaptorQ large \object}
\label{large raptorq sec}

For the RaptorQ code specified in~\cite{RFC6330}, the maximum number of
supported source symbols is $56,403$, and thus the maximum supported source block
size $WS$ should satisfy $\lfloor WS/T \rfloor \le 56,403$, where $T$ is the \ssize.
With $T=1400$ bytes, i.e., around the typical IPv4 packet payload size, 
the maximum \dobjsz\ supported without partitioning the \object\ into more than one 
source block is around $80$ Mbytes.

For RaptorQ~\cite{RFC6330}, $N=2^{31}-1$ is a good choice, and thus
up to $2^{31}-1$ source blocks can be supported with the design described in
Section~\ref{large sopi sec}.  With $T$
set to the maximum supported \ssize\ of $2^{16} = 65,536$ bytes, 
the maximum size \object\ supported is around $8 \cdot 10^{18}$ bytes,
or around $8$ Exabytes.


\end{document}